\documentclass[11pt]{article}

\usepackage[sort,longnamesfirst]{natbib}

\usepackage{amsbsy,amsmath,amsthm,amssymb,graphicx,verbatim,color}
\usepackage{subfigure}
\usepackage{multirow}

\usepackage{geometry}
\newtheorem{theorem}{Theorem}
\newtheorem{proposition}{Proposition}
\newtheorem{corollary}{Corollary}
\newtheorem{lemma}{Lemma}

\theoremstyle{remark}

\newcommand{\lb}{\left\lbrace}
\newcommand{\rb}{\right\rbrace}

\newcommand{\tx}{\tilde{x}}
\newcommand{\ty}{\tilde{y}}

\newcommand{\sig}{\sigma}

\begin{document}

\title{Geometric Ergodicity \& Scanning Strategies for Two-Component Gibbs Samplers}

\author{ Alicia A. Johnson\footnote{ {\tt ajohns24@macalester.edu}} and Owen Burbank\footnote{Currently at: Epic Systems, Madison, WI 53711} \\
            Department of Mathematics, Statistics, and Computer Science\\
           Macalester College \\
       }
       \date{Draft: \today} 
\maketitle

\begin{abstract}
In any Markov chain Monte Carlo analysis, rapid convergence of the chain to its target probability distribution is of practical and theoretical importance.  A chain that converges at a geometric rate is {\it geometrically ergodic}.
In this paper, we explore geometric ergodicity for two-component Gibbs samplers which, under a chosen {\it scanning strategy}, evolve by combining one-at-a-time updates of the two components. 
We compare convergence behaviors between and within three such strategies: composition, random sequence scan, and random scan.  Our main results are twofold.  First, we establish that if the Gibbs sampler is geometrically ergodic under any one of these strategies, so too are the others.  Further, we establish a simple and verifiable set of sufficient conditions for the geometric ergodicity of the Gibbs samplers.  Our results are illustrated using two examples. \end{abstract}

\newpage

%










%

\section{Introduction}


Providing a framework for approximately sampling from complicated target probability distributions, Markov chain Monte Carlo (MCMC) methods facilitate
statistical inference in intractable settings.
Consider distribution  $\varpi$ with support on some general state space in $\mathbb{R}^d$.   
Implementation of the foundational Metropolis-Hastings MCMC algorithm for $\varpi$ requires full-dimensional draws from an approximating {\it proposal} distribution.  
However, in settings requiring MCMC, $\varpi$ is typically complicated or $d$ large.  Thus, constructing an appropriate proposal can be  prohibitively difficult.   In such cases, we might instead employ a {\it component-wise} strategy which updates $\varpi$ one variable or sub-block of variables at a time.  

The fundamental component-wise MCMC algorithm, the {\it Gibbs sampler} (GS), evolves by updating each {\it sub-block} or  {\it component} with draws from its conditional distribution given the other components.  
For example, suppose we block the variables of $\varpi$ into two components, $X \in  \mathbb{R}^{d_x}$ and $Y \in  \mathbb{R}^{d_y}$ where $d_x, d_y \ge 1$ and $d_x+d_y=d$.
Let $\pi(x,y)$ denote the corresponding two-component density admitted by $\varpi$ with respect to measure $\mu = \mu_x \times \mu_y$ and 
having  support $ \mathsf{X} \times \mathsf{Y} \subseteq \mathbb{R}^{d_x} \times \mathbb{R}^{d_y}$.   Then the GS Markov chain $\Phi := \{(X^{(0)},Y^{(0)}), (X^{(1)},Y^{(1)}), (X^{(2)},Y^{(2)}), \ldots \}$ evolves by drawing updates of $X$ and $Y$ from the conditional densities $\pi(x|y) := \pi(x,y)/\int \pi(x,y) \mu_x(dx)$ and $\pi(y|x) := \pi(x,y)/\int \pi(x,y) \mu_y(dy)$, respectively.  
Let $P^n((x,y), A)$ denote the corresponding $n$-step {\it Markov transition kernel} where for state $(x,y) \in \mathsf{X}\times \mathsf{Y}$, set $A$ in the Borel $\sigma$-algebra $\mathcal{B}$ on $\mathsf{X}\times\mathsf{Y}$, and $n,i \in Z^{+}$ 
\[
P^n((x,y),A) = \text{Pr}\left(\left(X^{(i+n)},Y^{(i+n)}\right) \in A \; \vline \; \left(X^{(i)},Y^{(i)}\right) = (x,y) \right) \; .
\]
When the GS is Harris ergodic (ie.~$\pi$ irreducible, aperiodic, and Harris recurrent with invariant density $\pi$ \citep{meyn:twee:1993}), $\Phi$ converges to $\varpi$ in {\it total variation distance}.  That is,
$\parallel P^n((x,y), \cdot) - \varpi(\cdot) \parallel : =  \sup_{A \in \mathcal{B}} | P^n((x,y),A) - \varpi(A)| \to 0$  as $n \to \infty$.
Understanding the rate of this convergence is paramount in evaluating the quality of Markov chain output.  To this end,  we say $\Phi$ is {\it geometrically ergodic} if there exist some function $M : \mathsf{X}\times\mathsf{Y} \to \mathbb{R}$ and constant $t \in (0,1)$ for which
\begin{equation}\label{eq:geoerg}
\parallel P^n((x,y), \cdot) - \varpi(\cdot) \parallel \le t^n M(x,y) \;\;\; \text{ for all } (x,y) \in \mathsf{X} \times \mathsf{Y} \; .
\end{equation}
A geometric convergence rate is crucial for several reasons, not least of which is achieving effective simulation results in finite time.  Perhaps most importantly, geometric ergodicity ensures that the same tools used for evaluating estimators in the independent and identically distributed sampling setting also exist in the GS setting.
Specifically, suppose we wish to calculate $E_{\varpi}(g) : =  \iint  g(x,y)\pi(x,y) \mu_x(dx)\mu_y(dy)$ for $g: \mathsf{X}\times\mathsf{Y} \to \mathbb{R}$.   Under Harris ergodicity, the Monte Carlo estimate $\overline{g}_n := \sum_{i=0}^{n-1} g\left(X^{(i)}, Y^{(i)} \right)$  converges to $E_{\varpi}(g)$ with probability one as $n \to \infty$.  
Further, if $E_{\varpi}|g|^{2+\delta} < \infty$ for some $\delta > 0$, geometric ergodicity ensures the existence of a Markov chain Central Limit Theorem (CLT)
\begin{equation}\label{eq:clt}
\sqrt{n}(\overline{g}_n - E_{\varpi}(g)) \stackrel{d}{\to} N\left(0, \sigma_g^2 \right) \;\; \text{ as } n \to \infty 
\end{equation}
for $0 < \sigma_g < \infty$ \citep{jone:2004}.  Under these same conditions, batch means, spectral methods and regenerative
simulation methods provide asymptotically valid Monte Carlo standard errors for $\overline{g}_n$, $\hat{\sigma}_g/\sqrt{n}$ \citep{atch:2011,fleg:jone:2010,hobe:jone:pres:rose:2002,
  jone:hara:caff:neat:2006}.
In turn, we can rigorously assess the accuracy of $\overline{g}_n$ and determine a sufficient simulation length $n$ 
\citep{fleg:hara:jone:2008, fleg:jone:2010}.

Accordingly, our goal is to explore geometric ergodicity for the two-component GS.
Studying this special case is a crucial first step in understanding convergence for GS with multiple components and has many practical applications. For instance, two-component GS serves as the foundation of data augmentation methods and can be used to explore such practically relevant models as the Bayesian general linear model in \cite{john:jone:2010}.
Our work in this GS setting is twofold. 
First, we explore convergence behavior under three different GS {\it scanning strategies}: composition, random sequence scan, and random scan.  For one, we establish that if the GS under any one of these strategies is geometrically ergodic, they all are.  
These results fill in gaps left by 
\cite{john:jone:neat:2013} who explore convergence of component-wise samplers in the general setting.  
Second, we provide a simple set of sufficient conditions for the geometric ergodicity of the GS.  
Such conditions exist for selected model-specific settings (see, for example, \cite{diac:etal:2008b,diac:etal:2008,hobe:geye:1998,jone:hobe:2004, robe:rose:1998c, john:jone:2010}).
However, there is a lack of verifiable conditions that can be utilized in general settings.  For example, though \cite{gema:gema:1984} and \cite{liu:wong:kong:1995} provide general, sufficient conditions for geometric ergodicity, Geman and Geman only consider
GS on finite state spaces and the conditions in Liu et al.~are admittedly difficult to establish in practice.  Further, in their Proposition 1, \cite{tan:jone:hobe:2011} note the need for a {\it drift condition}, but stop short of providing guidance on how to construct such a condition.

We begin in Sections \ref{sec:background} and  \ref{sec:ge} with an overview of GS and geometric ergodicity, respectively.   In Section \ref{sec:GSgeo} we explore geometric convergence  of the GS under different scanning strategies and,  in Section \ref{sec:recipes}, present sufficient conditions for geometric ergodicity.  Finally, we illustrate our results using two examples in Section \ref{sec:ex}.

\section{Background}\label{sec:bg}
\subsection{The Gibbs Sampler}\label{sec:background}
Consider the two-component GS Markov chain $\Phi := \{(X^{(0)},Y^{(0)}), (X^{(1)},Y^{(1)}), (X^{(2)},Y^{(2)}), \ldots \}$.  In general, $\Phi$ evolves by drawing $X$ and $Y$ updates from the full conditional densities $\pi(x|y)$ and $\pi(y|x)$, respectively.  However, the order and frequency of component-wise updates depends on the chosen scanning strategy. Three fundamental strategies are  composition (CGS), random sequence scan (RQGS), and random scan (RSGS).   

First, in every iteration of the CGS, $X$ and $Y$ are updated in a fixed, predetermined order.  Without loss of generality, we assume throughout that $X$ is updated first.  Thus the CGS Markov kernel $P_{CGS}$ admits Markov transition density (Mtd) 
\[
\begin{split}
k_{CGS}((x,y),(x',y')) & = \pi(x'|y)\pi(y'|x') \; .\\
\end{split}
\]
Specifically, 
$P_{CGS}((x,y),A) 
=  \iint\limits_A  k_{CGS}((x,y),(x',y')) \mu_x(dx')\mu_y(dy')$.
RQGS also updates both $X$ and $Y$ in each iteration.  However, the update order is randomly selected according to {\it sequence selection probability} $q\in (0,1)$.  Letting $q$ be the probability of updating $X$ first and $1-q$ the probability of $Y$ first, the RQGS Markov kernel $P_{RQGS,q}$ admits Mtd 
\[
k_{RQGS,q}((x,y),(x',y')) = q\pi(x'|y)\pi(y'|x') + (1-q)\pi(y'|x)\pi(x'|y') \; .
\]  
Thus the RQGS is essentially a mixture of the two possible composition scan GS (that which first updates $X$ and that which first updates $Y$).  Moreover, when $q$ is close to 1, the RQGS behaves much like the CGS which first updates $X$.

Finally, unlike CGS and RQGS, RSGS randomly selects a single component for update in each iteration while fixing the other.  Letting {\it component selection probability} $p \in (0,1)$ be the probability of updating $X$ and $1-p$ the probability of updating $Y$, the RSGS Markov kernel $P_{RSGS,p}$ admits Mtd 
\[
\begin{split}
k_{RSGS,p}((x,y),(x',y')) 
& = p \pi(x'|y)\delta(y'-y) \\
&\hspace{.25in} + (1-p) \pi(y'|x) \delta(x'-x) \\
\end{split}
\]
where $\delta$ is Dirac's delta.  Thus for $p$ close to 1, the RSGS will produce many $X$ updates  but just as many repeat copies of $Y$.
The opposite is true for $p$ close to 0.

Though CGS may be more familiar to readers, there are certain advantages to considering RSGS and RQGS.  For example, it is easy to show that RSGS is reversible with respect to $\pi$ for all $p$ and RQGS is reversible for $q=1/2$.  This, among other advantages, weakens the conditions for a CLT \citep{jone:2004}.  Specifically, \eqref{eq:clt} holds if $g$ has a finite second moment.  In the reversible setting, we can also compare and measure the quality of the GS through Peskun ordering and variance bounding properties \citep{robe:rose:2008}.  


\subsection{Establishing Geometric Ergodicity}\label{sec:ge}
Studying convergence properties of the GS requires a few definitions.  First, let $P$ denote a generic GS Markov kernel (CGS, RQGS, or RSGS) with Mtd $k$.  
We say  a {\it drift condition} holds if there exist some {\it drift function} $V: \mathsf{X} \times \mathsf{Y} \to [1,\infty)$, {\it drift rate} $0 < \lambda <1$, and constant $b < \infty$ such that
\begin{equation}\label{eq:drift}
PV(x,y) \le \lambda V(x,y) + b \;\; \text{ for all } (x,y) \in \mathsf{X} \times \mathsf{Y} 
\end{equation}
where, here applied to a function, $P$ acts as an operator with
\begin{equation}\label{eq:operator}
\begin{split}
PV (x,y) & := E\left[V \left(X^{(t+1)},Y^{(t+1)}\right) \; \vline \; \left(X^{(t)},Y^{(t)}\right) = (x,y) \right] \\
& =  \iint\limits_{\mathsf{X}\times\mathsf{Y}}V(x',y') k((x,y),(x',y')) \mu_x(dx')\mu_y(dy')  \; .\\
\end{split}
\end{equation}
We say $V$ is {\it unbounded off compact sets} if the set $D_d := \{(x,y): V(x,y) \le d\}$ is compact for all $d > 0$.  
Together, if \eqref{eq:drift} holds for some $V$ that is unbounded off compact sets, $\Phi$ will ``drift" toward values of $(x,y)$ for which $V(x,y)$ is small (ie.~close to 1).  (See \cite{jone:hobe:2001} for an in-depth discussion.) The rate of this drift is captured by $\lambda$;  the smaller the $\lambda$, the quicker the drift.  Thus, smaller $\lambda$ are loosely indicative of quicker convergence.  In fact, Markov chain drift  is a sufficient condition for geometric ergodicity.  The following proposition follows 
from Lemma 15.2.8 and Theorems 6.0.1 and 15.0.1 of \cite{meyn:twee:1993}.
\begin{proposition}\label{prop:geoerg}
Suppose the support
of $\pi$ has non-empty interior and Markov chain $\Phi$ is Harris ergodic and Feller, that is, for any open set $O \in \mathcal{B}$
\[
\liminf_{(x_n,y_n) \to (x,y)} P((x_n,y_n), O) \ge P((x,y), O) \hspace{.15in} \text{ for } (x_n,y_n), (x,y) \in \mathsf{X}\times \mathsf{Y} \; .
\]
Then, if drift condition \eqref{eq:drift} holds for some $V$ that is unbounded off compact sets,  $\Phi$ is geometrically ergodic. 
\end{proposition}

\section{Geometric Ergodicity of the Gibbs Sampler}

Our main goal is to explore geometric ergodicity within and between the CGS, RQGS, and RSGS.  To this end, 
we investigate the impact of GS scanning strategy on achieving geometric convergence in Section \ref{sec:GSgeo} 
and provide sufficient and verifiable conditions for the geometric ergodicity of the GS in Section \ref{sec:recipes}.  

\subsection{Geometric Ergodicity Under Different Scanning Strategies}\label{sec:GSgeo}

GS convergence rates depend on both target distribution $\varpi$ and scanning strategy.
Though different scanning strategies can produce chains with differing asymptotic behaviors, the common building blocks of the CGS, RQGS, and RSGS (namely, the full conditional distributions used for component-wise updates) suggest there should also be links among their convergence properties.  
In this section we address three main questions: (Q1) Does geometric ergodicity of any one of CGS, RQGS, or RSGS guarantee the same for the others?; (Q2) Does geometric ergodicity of the RQGS using sequence selection probability $q$ guarantee the same for all other selection probabilities?; and (Q3) Does geometric ergodicity of the RSGS using component selection probability $p$ guarantee the same for all other selection probabilities?  
Under a  set of conditions on target density $\pi$ and the GS,
the answer to all three of these questions is YES.  
We call this set of conditions {\it Assumption $\mathcal{A}$} which is satisfied if
\begin{enumerate}
\item[(a)]  the CGS, RQGS, and RSGS are Harris ergodic;
\item[(b)] the support of $\pi$ has non-empty interior with respect to $\mu_x\times \mu_y$; and 
\item[(c)]  
for all $(x,y), \; (x_n,y_n) \in \mathsf{X}\times\mathsf{Y}$ 
\begin{equation}\label{eq:c}
\pi\left(y \; \vline \; \liminf_{n\to\infty} x_n \right) \le \liminf_{n\to\infty} \pi\left(y \; \vline \;  x_n \right) 
\hspace{.15in} \text{ and } \hspace{.15in}
\pi\left(x \; \vline \; \liminf_{n\to\infty} y_n \right) \le \liminf_{n\to\infty} \pi\left(x \; \vline \;  y_n \right) \; .
\end{equation}

\end{enumerate}

We believe that {\it Assumption $\mathcal{A}$} does not significantly restrict the usefulness of our results.  First, (a) and (b) are required by Proposition \ref{prop:geoerg} where (a) is also a standard assumption for any Markov chain.  Further, (b) and (c) are satisfied by a wealth of target densities on general state spaces explored by GS in practice.  For example, (b) should hold in most cases where the maximal irreducibility measure for $\Phi$ is Lebesgue and (c) holds when $\pi$ is continuous.  In fact, (c) is simply a sufficient condition for the GS to be Feller.  A proof of this Lemma is given in the appendix.
\begin{lemma}\label{lem:feller}
If \eqref{eq:c} holds for all $(x,y), \; (x_n,y_n) \in \mathsf{X}\times\mathsf{Y}$, then the CGS, RQGS, and RSGS are Feller.
\end{lemma}
With these assumptions in place, the following theorem from \cite{john:jone:neat:2013} will be critical in addressing (Q1), (Q2), and (Q3).
\begin{theorem}\label{thm:C}
Under {\it Assumption $\mathcal{A}$},  if CGS is geometrically ergodic, so are the RQGS for all sequence selection probabilities $q$ and the RSGS for all component selection probabilities $p$.
\end{theorem}
This result captures a clear connection between the convergence behavior of the CGS, RQGS, and RSGS.  However, it fails to address (Q2) and (Q3).  It also only provides  an incomplete look into (Q1).  Specifically, Theorem \ref{thm:C} proves that geometric ergodicity of RQGS and RSGS follow from that of the CGS, but not the converse.  
We fill in these gaps below, starting with an exploration of the RQGS.  All proofs can be found in the appendix.

\begin{theorem}\label{thm:Q}
Under {\it Assumption $\mathcal{A}$},  if RQGS is geometrically ergodic for some sequence selection probability $q \in (0,1)$, then so is the CGS.
\end{theorem}

Corollary \ref{cor:Q} follows directly from Theorems \ref{thm:C} and \ref{thm:Q}.
\begin{corollary} \label{cor:Q}
Under {\it Assumption $\mathcal{A}$},  if RQGS is geometrically ergodic for some sequence selection probability $q \in (0,1)$, it is geometrically ergodic for {\it all} $q \in (0,1)$.
\end{corollary}

The results of Theorem \ref{thm:Q} and Corollary \ref{cor:Q} are, perhaps, intuitive.  It is well known that the two-component CGS updating $X$ then $Y$ has the same convergence rate as that updating $Y$ then $X$.  Thus, if some mixture of these samplers (ie.~RQGS) is geometrically ergodic, so too should be the individual components.  Further, these results confirm that if {\it some} mixture of the geometrically ergodic CGS is geometrically ergodic, then {\it all} possible mixtures are geometrically ergodic.  Next, we establish similar results for the RSGS.

\begin{theorem} \label{thm:S}
Under {\it Assumption $\mathcal{A}$},  if RSGS is geometrically ergodic for some component selection probability $p$, then so is the CGS.
\end{theorem}

Corollary \ref{cor:S} follows directly from Theorems \ref{thm:C} and \ref{thm:S}.

\begin{corollary} \label{cor:S}
Under {\it Assumption $\mathcal{A}$},  if RSGS is geometrically ergodic for some component selection probability $p \in (0,1)$, it is geometrically ergodic for {\it all} $p \in (0,1)$.
\end{corollary}

Consider Theorem \ref{thm:S}.  It is natural to believe that if the RSGS converges at a geometric rate by updating a single component in each iteration,  so too should the CGS which updates both components in each iteration.  The result of Corollary \ref{cor:S}, on the other hand, might be more surprising.  In its extreme, this corollary asserts that if a RSGS updating $X$ with high frequency ($p \approx 1$) is geometrically ergodic, so is the RSGS updating $X$ with low frequency ($p \approx 0$).  In other words, if a chain converges quickly by spending the majority of its effort exploring one component of the state space while getting stuck in the other, so too will it converge quickly by spending its effort exploring the other component of the state space.

Finally, combining the above theorems establishes
 Theorem \ref{thm:all}, our main result.
\begin{theorem} \label{thm:all}
Under {\it Assumption $\mathcal{A}$},  suppose any one of the CGS, RQGS, or RSGS are geometrically ergodic.  Then so are the others, regardless of RQGS and RSGS selection probabilities $q$ and $p$, respectively.
\end{theorem}

It is important to note that Theorem \ref{thm:all} does not assert that the CGS, RQGS, and RSGS converge at the same rate.     In fact, if these samplers satisfy \eqref{eq:geoerg} for different $t$ and $M(\cdot)$, their exact convergence rates, though all geometric, may significantly differ.  The same is true within the RQGS and RSGS under different selection probabilities $q$ and $p$, respectively.  Thus choice of scanning strategy and choice of $q$ and $p$ within RQGS and RSGS may impact the empirical performance of a finite GS simulation.   
Assuredly, whether the geometric convergence is relatively fast or slow, the existence of a Markov chain CLT \eqref{eq:clt} provides a means for rigorously assessing the quality of MCMC inference.
Though not the focus of this paper, we explore the impact of  scanning strategy on finite simulation quality with a short study in Section \ref{sec:nn}.  For a more in-depth discussion of the impact of $p$ in RSGS, please see \cite{levi:etal:2005, levi:case:2006, liu:wong:kong:1995} and see \cite{john:jone:neat:2013} for further discussion of comparisons between CGS, RSGS, and RQGS.

\subsection{Sufficient Conditions for Geometric Ergodicity}\label{sec:recipes}

We end this section with a simple set of sufficient conditions for the geometric ergodicity of the GS.  
By no means are these conditions exhaustive.  Our goal is to merely provide guidance for those new to establishing geometric ergodicity. A proof of Theorem \ref{thm:recipes} is provided in the appendix.  We recommend inspection of this proof to develop intuition for establishing geometric ergodicity.  

\begin{theorem}\label{thm:recipes}

Suppose {\it Assumption $\mathcal{A}$} holds and that there exist functions $f: \mathsf{X} \to [1,\infty)$ and $g: \mathsf{Y} \to [1,\infty)$ and constants 
$j,k,m,n > 0$ such that $jm < 1$ and 
\begin{equation}\label{eq:fg}
\begin{split}
E[f(x)|y] & \le j g(y) +  k \\
E[g(y)|x] & \le m f(x) + n \; .\\
\end{split}
\end{equation}
Then if $C_d := \{y: g(y) \le d\}$ is compact for all $d>0$, the CGS, RQGS, and RSGS are geometrically ergodic.  

\end{theorem}

Lemma \ref{lem:recipes} follows directly from the proof of Theorem \ref{thm:recipes}.
\begin{lemma}\label{lem:recipes}
Under the assumptions of Theorem \ref{thm:recipes},  CGS, RQGS, and RSGS drift conditions \eqref{eq:drift} can be constructed as follows.  For CGS, 
\[
P_{CGS}V_{CGS}(x,y)  \le \lambda_{CGS} V_{CGS}(x,y) + b_{CGS} 
\]
holds for $V_{CGS}(x,y) = g(y)$, $jm \le \lambda_{CGS}  < 1$, and $b_{CGS} = mk + n$.  For RQGS with sequence selection probability $q$, define
\[
v_{RQGS,q} = \frac{(2q-1)jm + \sqrt{jm(jm + 4q(1-q)(1-jm))}}{2(1-q)m} \;.
\]
Then
\[
P_{RQGS,q}V_{RQGS}(x,y) \le \lambda_{RQGS} V_{RQGS}(x,y) + b_{RQGS} 
\]
holds for $V_{RQGS}(x,y) = f(x) + v_{RQGS,q}g(y)$,  
\[
\begin{split} 
b_{RQGS} & = q[k + v_{RQGS,q}(mk+n)] + (1-q)[v_{RQGS,q}n + (jn+k)] \;, \;\; \text{ and } \\
(1-q)(j+v_{RQGS,q})m & = \frac{1}{2}(jm + \sqrt{jm[jm + 4q(1-q)(1-jm)]})   \le \lambda_{RQGS}  < 1\;. \\
\end{split}
\]
 Finally, for RSGS with component selection probability $p$, define 
\[
v_{RSGS,p} = \frac{(2p-1)jm + \sqrt{ 1- 4p(1-p)(1-jm)}}{2(1-p)m} \; .
\]
Then
\[
P_{RSGS,p}V_{RSGS}(x,y) \le \lambda_{RSGS} V_{RSGS}(x,y) + b_{RSGS}
\]
holds for $V_{RSGS}(x,y) = f(x) + v_{RSGS,p}g(y)$, $b_{RSGS} = pk + (1-p)v_{RSGS,p}n$,  and
\[
(1-p)(1+ v_{RSGS,p}m) = \frac{1}{2}(1 + \sqrt{1 - 4p(1-p)(1-jm)})  \le \lambda_{RSGS} < 1 \; .
\]

\end{lemma}

In constructing the functions $f$ and $g$ required by Theorem \ref{thm:recipes}, keep in mind the following guidelines.  First, the conditional expectations of $f$ and $g$ must maintain a cyclic-type relationship \eqref{eq:fg}. Functions satisfying this requirement can often be found by exploring lower moments of the conditional distributions of $X|Y$ and $Y|X$.  Next, Lemma \ref{lem:recipes} demonstrates that CGS, RQGS, and RSGS drift functions can each be constructed as linear combinations of $f$ and $g$ (further evidence of systematic connections between their convergence behaviors).  Recall that the Markov chain will drift toward values for which the drift function is small.  Thus attention should be focused on functions $f$ and $g$ that take on small values in the center of the state space where density $\pi$ is largest.

These concerns regarding $f$ and $g$ are specific to Theorem \ref{thm:recipes} which  presents a single, but not exhaustive, set of sufficient conditions for geometric ergodicity.  In turn, the drift conditions and drift rates provided by Lemma \ref{lem:recipes} are not unique.  However, 
as smaller drift rates are loosely indicative of faster convergence,  $\lambda_{CGS}$, $\lambda_{RQGS}$, and $\lambda_{RSGS}$ provide interesting insight into the convergence relationships between and within the CGS, RQGS, and RSGS.  
%
To this end, first notice the dependence of RQGS drift rate $\lambda_{RQGS}$ on $q$.  Mainly, $\lambda_{RQGS}$ increases as $q$ approaches $1/2$ and  
converges to its lower bound, $\lambda_{CGS}=jm$, as $q$ approaches 0 or 1. 
This suggests that the RQGS drift is quickest when one of the update orders
is strongly favored over the other, that is, when RQGS behaves like CGS.   Similarly, $\lambda_{RSGS}$ is minimized (hence drift is quickest) when $p = 1-p = 1/2$, that is, when updates of $X$ and $Y$ are roughly balanced.  It is in this setting that the RSGS behaves most like CGS.  Finally, 
we can compare the CGS, RQGS, and RSGS drift rates.
Indeed, since the RSGS requires at least two iterations to update both $X$ and $Y$ whereas the CGS and RQGS require only one, a more fair comparison might be among 
$\lambda_{CGS}$, $\lambda_{RQGS}$, and $\lambda_{RSGS}^2$, the drift rate corresponding to the two-step RSGS drift condition:
\[
\begin{split}
P_{RSGS,p}^2 V_{RSGS}(x,y) & = P_{RSGS,p}( P_{RSGS,p} V_{RSGS}(x,y)) \\
& \le P_{RSGS,p}(\lambda_{RSGS} V_{RSGS}(x,y) + b_{RSGS}) \\
& \le \lambda_{RSGS}^2 V_{RSGS}(x,y) + b_{RSGS}(1 + \lambda_{RSGS}) \; . \\
\end{split}
\]
Given the definitions in Lemma \ref{lem:recipes}, it follows that $\lambda_{CGS} < \lambda_{RQGS} < \lambda_{RSGS}^2 < \lambda_{RSGS}$.  Though this seems to suggest that the CGS converges quicker than the RQGS which converges quicker than the RSGS (both the original and two-step versions), we again caution against placing too much importance on interpreting this single set of possible $\lambda$.


\section{Examples}\label{sec:ex}

We illustrate our results using two examples.  
The first is a toy example of GS for a Normal-Normal model.  
Included is a simulation study which explores the impact of scanning strategy on the empirical quality of finite GS for this model.
The second considers GS for a special case of the Bayesian general linear model studied by \cite{john:jone:2010}.  This model is practically relevant in that inference for the corresponding Bayesian posterior distribution requires MCMC.  

\subsection{A Normal-Normal Model}\label{sec:nn}
Let $X=(X_1,X_2,\ldots,X_N) \in \mathbb{R}^N$ be an independent, identically distributed sample such that  $X_i|Y \sim N(Y, \theta^2)$ for each $i$ and $Y \in \mathbb{R}$ follows a $N(0, \tau^2)$ distribution.  Thus, the joint distribution of $(X,Y)$ is multivariate Normal with
\begin{equation}\label{eq:nn}
\left(\begin{array}{c} X \\ Y \end{array} \right) \sim N_{N+1}\left(\left(\begin{array}{c} 0_N \\ 0 \end{array} \right), \; 
\left(\begin{array}{cc} \theta^2 I_N + \tau^2 1_N1_N^T &  \tau^2 1_N \\  \tau^2 1_N^T & \tau^2 \end{array} \right) \right) 
\end{equation}
where $I_N$ is the $N$-dimensional identity matrix and  $0_N$ and $1_N$ are $N$-dimensional vectors of zeroes and ones, respectively.  

Inference for the Normal-Normal model does not require MCMC.  However, this model provides a nice setting in which to illustrate our results for two-component GS.  Let $\Phi = \lb\left(X^{(i)},Y^{(i)}\right) \rb_{i=0}^{\infty}$ be the GS chain which evolves by drawing from the conditional distributions \[
\begin{split}
X|Y &  \sim N(Y1_N, \theta^2 I_N) \\
Y | X & \sim N\left(\frac{\tau^2}{N\tau^2 + \theta^2}\sum_{i=1}^NX_i, \; \frac{\theta^2\tau^2}{N\tau^2 + \theta^2} \right) \\
\end{split}
\]
with first and second moments 
\[
\begin{split}
E(X_i | Y) &  = Y \\
E(X_i^2|Y) & = \theta^2+Y^2\\
E(Y|X) & = \frac{\tau^2}{N\tau^2 + \theta^2}\sum_{i=1}^NX_i\\
E(Y^2|X) & = \frac{\theta^2\tau^2}{N\tau^2 + \theta^2} + \left(\frac{\tau^2}{N\tau^2 + \theta^2}\right)^2 \left( \sum_{i=1}^NX_i\right)^2 \; .
\end{split}
\]The GS and Normal-Normal density clearly meet the conditions of {\it Assumption $\mathcal{A}$}:
the GS is Harris ergodic, the support of the Normal-Normal density is $\mathbb{R}^{N+1}$ which has non-empty interior with respect to Lebesgue measure, and the density is continuous hence satisfies condition (c).
Thus to establish geometric ergodicity for the CGS, RQGS, and RSGS we need only find functions $f$ and $g$ that satisfy the conditions of Theorem \ref{thm:recipes}.  
Per the discussion following Lemma \ref{lem:recipes}, this choice can be guided by the lower moments of the full conditionals.  Further, $f$ and $g$ should be small for values near the center of the state space.  To this end, we know that the Normal conditional distributions of $X|Y$ and $Y|X$ have areas of higher density near the values of $0_N$ and $0$, respectively.  With these guidelines in mind, consider defining 
\[
f(X) = \left(\sum_{i=1}^N X_i\right)^2  + 1
\hspace{.2in} \text{ and } \hspace{.2in} 
g(Y) = Y^2 + 1
\]
where 1 is added to $f$ and $g$ to ensure $f,g \ge 1$.  Note that these satisfy the requirement that $f$ and $g$ be small for values near $0_N$ and $0$, respectively.  Further, from the above conditional moments, it is straightforward to show that $f$ and $g$ satisfy \eqref{eq:fg} with $j=N^2$, $k=N\theta^2+1$,
\[
m = \left(\frac{\tau^2}{N\tau^2+\theta^2}\right)^2
\hspace{.2in} \text{ and } \hspace{.2in} 
n = \frac{\theta^2\tau^2}{N\tau^2+\theta^2} +1 \; .
\]
Finally, $C_d := \{y: g(y) \le d\} = [-\sqrt{d-1}, \sqrt{d-1}]$, hence is compact for all $d>0$.  It follows from Theorem \ref{thm:recipes} that the CGS, RQGS, and RSGS for the Normal-Normal model are geometrically ergodic.

Though the CGS, RQGS, and RSGS are each geometrically ergodic, their exact convergence rates may differ.
We explore these discrepancies and their impact on finite sample empirical performance  by comparing 
the CGS, RQGS for $q \in \{0.10, 0.25, 0.50, 0.75, 0.90\}$, and RSGS for $p \in \{0.10, 0.25, 0.50, 0.75, 0.90\}$ within  two different parameter settings:
\begin{center}
\begin{tabular}{|c||l|l|l||l|l|l|l|}
\hline\noalign{\smallskip}
Setting & $N$ & $\theta^2$  & $\tau^2$  & Var($X_i$) & Var($Y$) & Cor($X_i,X_j$) & Cor($X_i,Y$) \\
\noalign{\smallskip}\hline\noalign{\smallskip}
(1)        & 10  &  1                & 1              & 2  & 1       & 0.5  & 0.707 \\
(2)        & 10  &  1                & 0.1           & 1.1 & 0.1  & 0.091 & 0.302 \\
\noalign{\smallskip}\hline
\end{tabular}
\end{center}
where  the variance and correlation coefficients follow from \eqref{eq:nn}.  
Before presenting our results, we remind the reader that GS convergence and performance depend both on scanning strategy and target distribution.  Thus the comparisons we make between the GS below should not be generalized far beyond the specific Normal-Normal settings studied here.

To begin, consider one long run of each GS in both settings.  Starting from $\left(X^{(0)},Y^{(0)}\right) = 0_{11}$, we independently ran the CGS and RQGS for $10^5$ iterations and RSGS for $2*10^5$ iterations since, again, RSGS requires at least twice as many iterations as the CGS and RQGS to obtain the same number of $X$ and $Y$ updates.
Trace plots of the final 1000 $Y$ iterations  for selected GS in Setting (1) are included in Figure \ref{fig:trace1}. The trace behavior is similar for the CGS and RQGS under the extreme settings of $q = 0.1$ and $q=0.9$.   On the other hand, as expected, the RSGS $Y$ sub-chain appears to mix more slowly than for CGS and RQGS both when $p=0.1$ ($Y$ is updated frequently) and, even worse, when $p=0.9$ ($Y$ is updated infrequently).
\begin{figure*}[h]\centering
 \subfigure[]{\includegraphics[height=3in,width=3in]{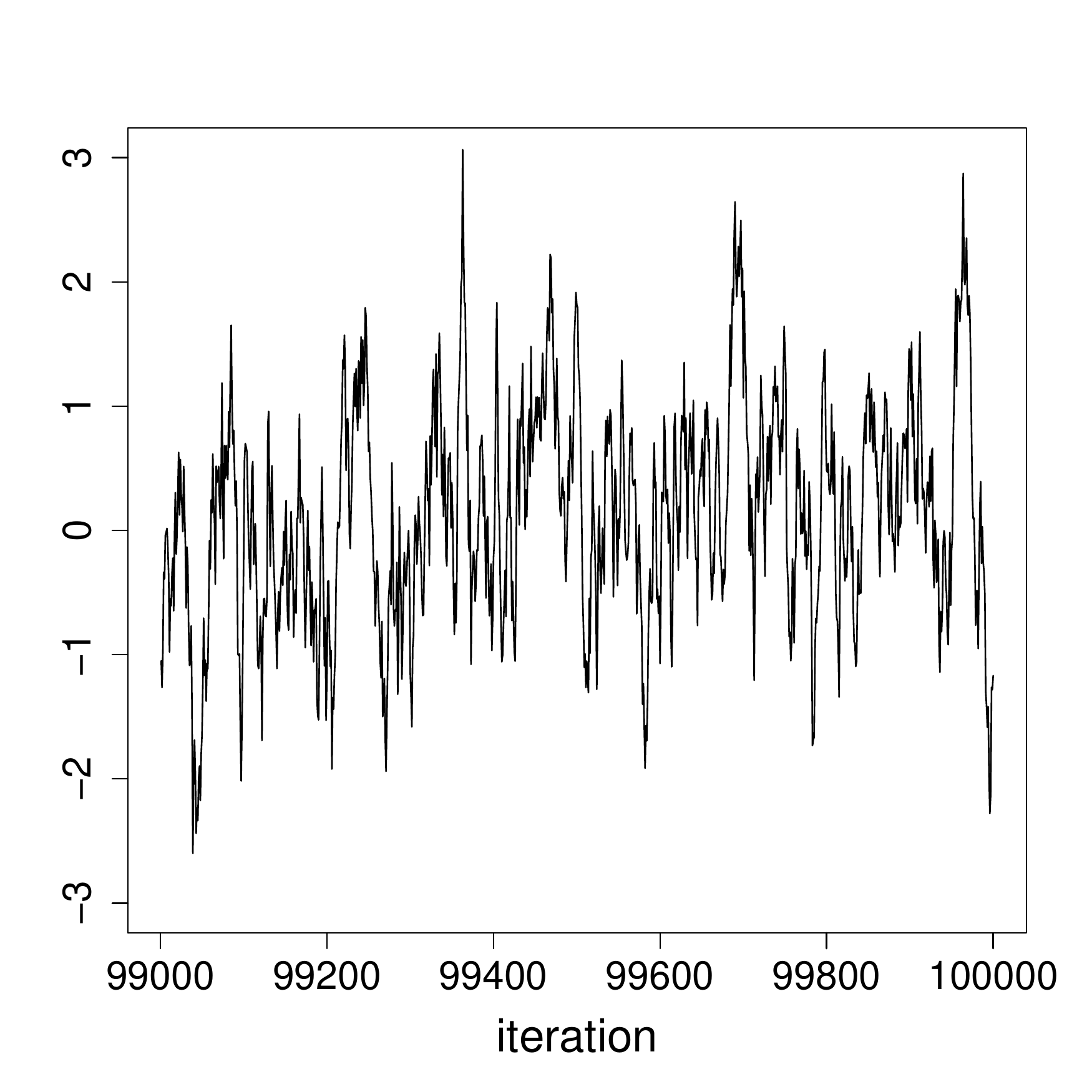}}
 \subfigure[]{\includegraphics[height=3in,width=3in]{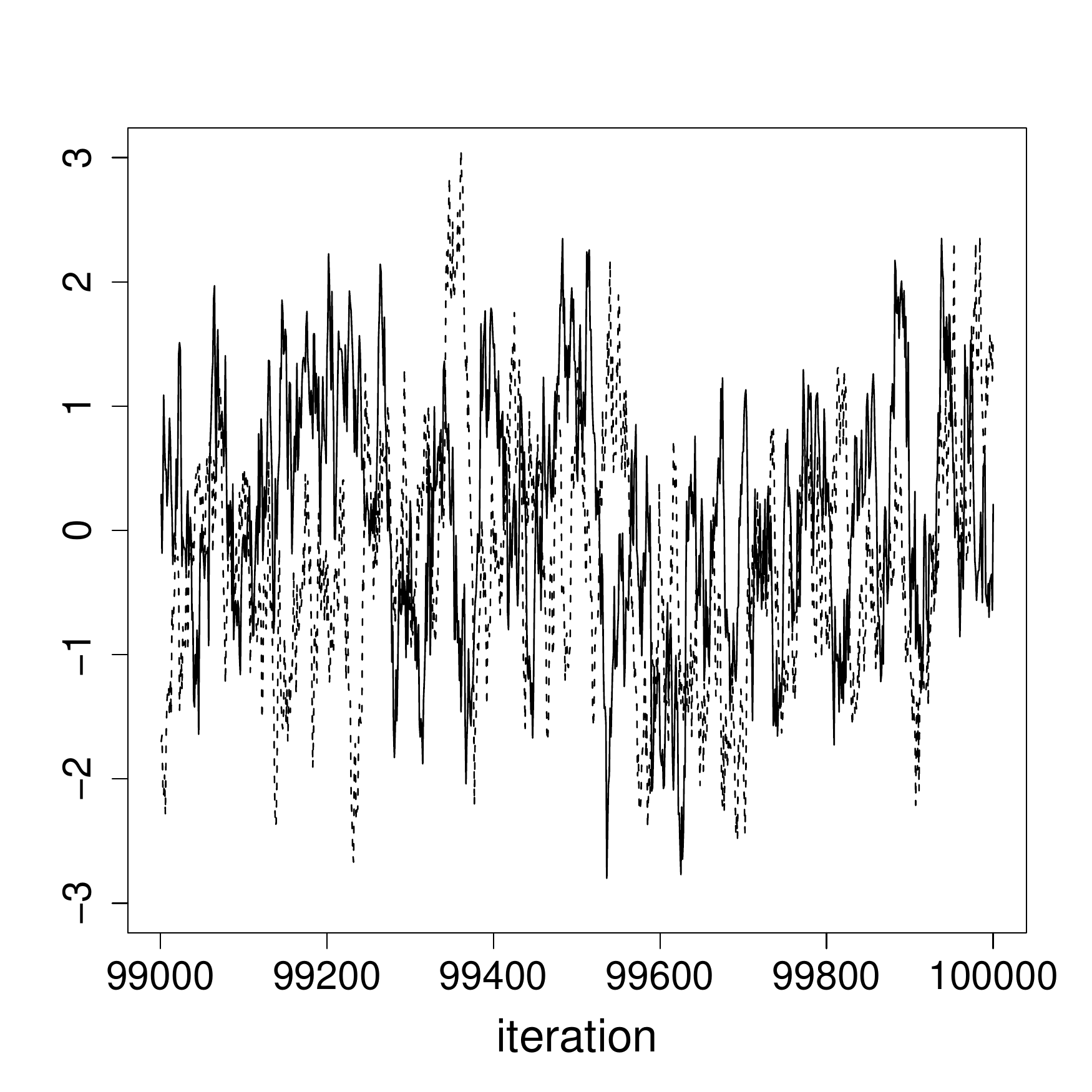}} 
 \subfigure[]{\includegraphics[height=3in,width=3in]{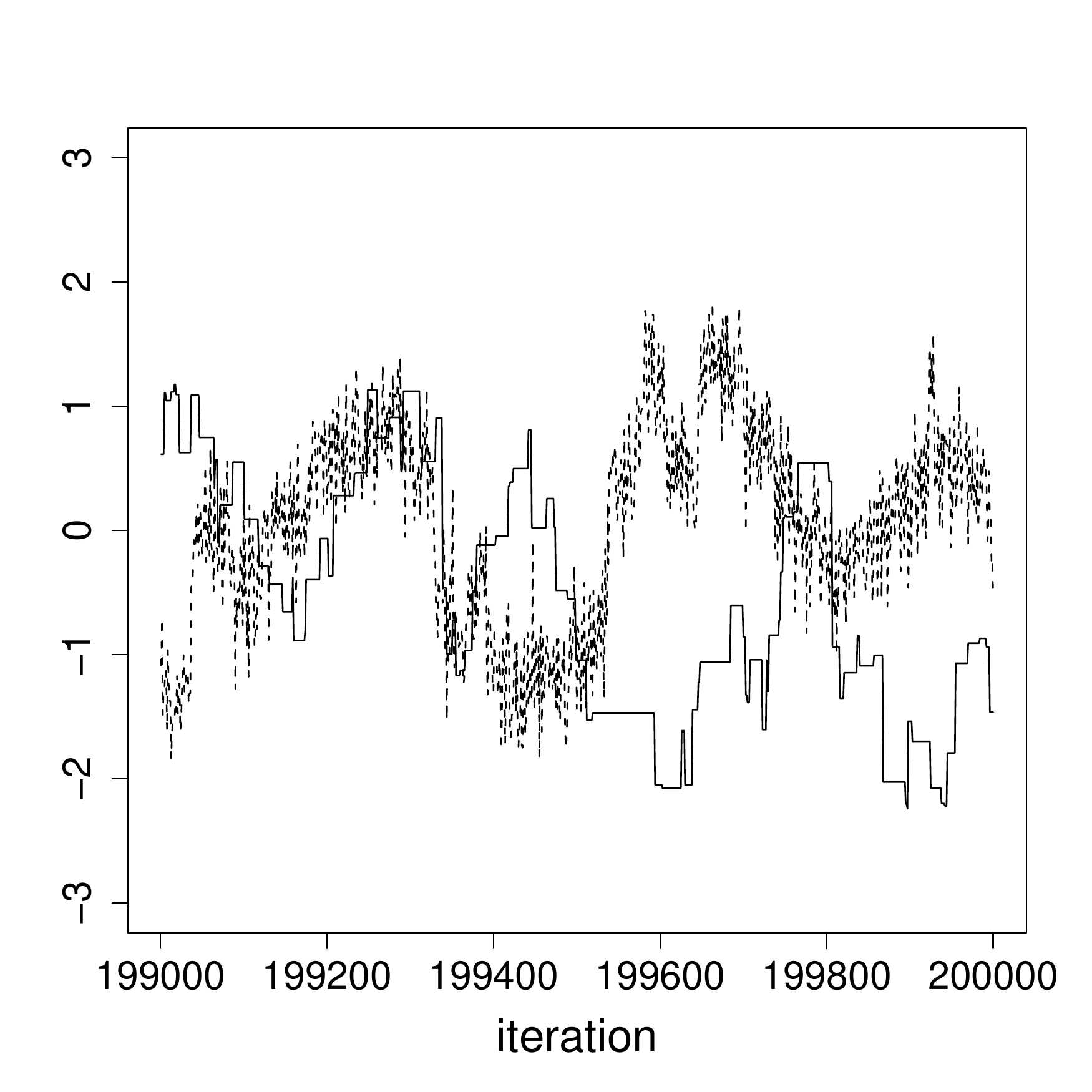}}
 \caption{Trace plots of GS for $Y$ in the Normal-Normal model of Section \ref{sec:nn}.  Shown are the last 1000 iterations of (a) $10^5$ CGS iterations, (b) $10^5$ RQGS iterations under $q=0.1$ (dashed) and $q=0.9$ (solid), and (c) $2*10^5$ RSGS iterations under $p=0.1$ (dashed) and $p=0.9$ (solid).}
 \label{fig:trace1}
\end{figure*}

More formally, we can compare GS efficiency  relative to the estimation of $E(Y)=0$. Since $E(Y^4) < \infty$, geometric ergodicity guarantees the existence of a Markov chain CLT for the Monte Carlo average $\overline{Y} = \sum_{i=0}^{n-1} Y^{(i)}$,
\[
\sqrt{n}(\overline{Y} - E(Y)) \stackrel{d}{\to} N(0,\sig^2_{\overline{Y}}) \; ,
\]
along with a consistent estimator of $\sig^2_{\overline{Y}}$, $\hat{\sig}^2_{\overline{Y}}$, via batch means methods.   
Thus an asymptotically valid 95\% confidence interval (CI) for $E(Y)$ can be calculated by 
\[
\overline{Y} \pm 1.960 \frac{\hat{\sig}_{\overline{Y}}}{\sqrt{n^*}}
\]
where $n^*$ denotes the MCMC simulation length ($n^*=10^5$ for CGS and RQGS and $n^*=2*10^5$ for RSGS).  Further,
the
integrated autocorrelation time
\[
\text{ACT} = \frac{\sigma^2_{\overline{Y}}}{\text{Var}(Y)} = \frac{\sigma^2_{\overline{Y}}/n^*}{\text{Var}(Y)/n^*} 
\]
can be consistently estimated by $\widehat{\text{ACT}} = \hat{\sigma}^2_{\overline{Y}}/\text{Var}(Y)$ and 
provides a measure of the GS efficiency relative to that of a random
sample from the Normal-Normal model.
Specifically,  the ACT indicates the number of GS iterations required for each random sample draw in order to achieve the same level of precision in estimating $E(Y)$.  

The 95\% confidence intervals and ACT's for each GS in Settings (1) and (2) can be found in Tables \ref{tab1} and \ref{tab2}, respectively.  Across GS scanning strategy, the CI half-widths and ACT's are larger in Setting (1) than in Setting (2).  This is to be expected since the variance of $Y$ and its correlation with $X$ are larger in Setting (1).
Further, comparisons of the GS empirical performances are similar within both settings and, interestingly, reflect the drift rate comparison discussion following Lemma \ref{lem:recipes}.  First, consider the comparisons between CGS, RQGS, and RSGS.  Nearly without exception, the CI half-widths and ACT's are substantially larger for the RSGS than the RQGS which are slightly larger than, but roughly comparable to, those for the CGS.  These patterns suggest that, relative to the estimation of $E(Y)$, CGS has a slight edge over RQGS and both are substantially more efficient than RSGS.  Next, consider the impact of selection probabilities $q$ and $p$ on the efficiencies of RQGS and RSGS, respectively.  
Within the RQGS, the CI half-widths and ACT's tend to decrease at a similar rate as $q$ nears 0 or 1.  In other words, the RQGS is more efficient when either one of the update orders is heavily favored over the other (ie.~when it behaves most like CGS).
On the other hand, RSGS efficiency appears to improve as $p$ nears 0.5, that is, when $X$ and $Y$ are updated at a roughly similar rate.
It is also interesting to note that, in both Settings (1) and (2), RSGS performs relatively better when $p$ is small (ie.~$Y$ is updated frequently) than when $p$ is large (ie.~$Y$ is updated infrequently).  Thus in this specific Normal-Normal setting, there does not seem to be an advantage to increasing the frequency of $Y$ updates as $\text{Var}(Y)$ (and $\text{Cor}(X_i,Y)$) increases.

Finally, we compare the quality of the Monte Carlo averages $\overline{Y}$ in estimating $E(Y)$ using mean squared error
\[
\text{MSE}(\overline{Y}) = E(\overline{Y} - E(Y))^2 = E(\overline{Y})^2 \; .
\]
To estimate the MSEs, for each GS within each parameter setting, we performed 1000 independent runs of either $10^4$ iterations each (CGS and RQGS) or $2*10^4$ iterations each (RSGS) and recorded the resulting independent estimates $\lb \overline{Y}^{(1)}, \overline{Y}^{(2)}, \ldots, \overline{Y}^{(1000)} \rb$.   From these, we estimate MSE by
\[
\widehat{\text{MSE}}(\overline{Y}) = \frac{1}{1000}\sum_{i=1}^{1000} \left( \overline{Y}^{(i)}\right)^2 \; .
\]
The results are reported in Tables \ref{tab1} and \ref{tab2} as MSE ratios relative to CGS
\[
\frac{\widehat{\text{MSE}}_{RQGS,q}(\overline{Y}) }{\widehat{\text{MSE}}_{CGS}(\overline{Y}) }
\hspace{.25in} \text{ and } \hspace{.25in} 
\frac{\widehat{\text{MSE}}_{RSGS,p}(\overline{Y}) }{\widehat{\text{MSE}}_{CGS}(\overline{Y}) } 
\]
where $\widehat{\text{MSE}}_{CGS}(\overline{Y})$ equals 0.00197 in Setting (1) and 0.0000295 in Setting (2).  
 Examination of the MSE ratios produces conclusions compatible with those from the CI half-widths and ACT's.  
 Mainly, CGS edges out RQGS (ie.~all ratios are greater than 1) and both are substantially more efficient than RSGS.  Further, RQGS is most efficient under $q$ values near 0 or 1 and RSGS is most efficient when $p=0.5$.

\begin{table}[H] \centering
 \caption{Summary of GS for the Normal-Normal model of Section \ref{sec:nn} under Setting (1).  
The 95\% CI's and ACT's are calculated from single independent runs of the GS. MSE ratios relative to the CGS MSE of 0.00197 are estimated from 1000 independent runs of each GS. Standard errors for the MSE ratios are in parentheses.  
}
 \label{tab1}
\begin{tabular}{|l||r|r||r|}
\hline\noalign{\smallskip}
Algorithm   & \multicolumn{1}{c}{95\% CI} & \multicolumn{1}{|c||}{ACT} & \multicolumn{1}{|c|}{MSE Ratio }\\
\noalign{\smallskip}\hline\noalign{\smallskip}
CGS              & 0.0067 $\pm$ 0.0273 & 19.289  &   \multicolumn{1}{|c|}{1} \\
\noalign{\smallskip}\hline\noalign{\smallskip}
\multicolumn{1}{|r||}{RQGS $\;\; q=0.10$} & 0.0054   $\pm$ 0.0287 & 21.241 & 1.217 (0.078)\\
 \multicolumn{1}{|r||}{$q=0.25$} & 0.0029  $\pm$ 0.0286 & 21.171  & 1.250 (0.081) \\
 \multicolumn{1}{|r||}{ $q=0.50$} & 0.0153  $\pm$ 0.0343 & 30.382  & 1.365 (0.086)\\
  \multicolumn{1}{|r||}{$q=0.75$} & 0.0167 $\pm$ 0.0289 & 21.554  & 1.206 (0.078)\\
 \multicolumn{1}{|r||}{ $q=0.90$} & -0.0188 $\pm$ 0.0274  & 19.459  & 1.158 (0.075)\\
\noalign{\smallskip}\hline\noalign{\smallskip}
 \multicolumn{1}{|r||}{RSGS  $\;\; p=0.10$} & -0.0132 $\pm$ 0.0567 & 166.592 & 5.701  (0.368) \\
 \multicolumn{1}{|r||}{ $p=0.25$} & 0.0223  $\pm$ 0.0425 & 93.622  & 2.705 (0.178) \\
 \multicolumn{1}{|r||}{ $p=0.50$} & 0.0177   $\pm$ 0.0379 & 74.480 & 2.178 (0.138) \\
 \multicolumn{1}{|r||}{ $p=0.75$} & -0.0485 $\pm$ 0.0440 & 100.246  & 2.662 (0.171)\\
 \multicolumn{1}{|r||}{ $p=0.90$} & 0.0033 $\pm$ 0.0589 & 179.441  & 6.279 (0.409)\\
\noalign{\smallskip}\hline
\end{tabular}
\end{table}

\begin{table}[H] \centering
 \caption{Summary of GS for the Normal-Normal model of Section \ref{sec:nn} under Setting (2).  
The 95\% CI's and ACT's are calculated from single independent runs of the GS. MSE ratios relative to the CGS MSE of 0.0000295 are estimated from 1000 independent runs of each GS. Standard errors for the MSE ratios are in parentheses. }
 \label{tab2}
\begin{tabular}{|l||r|r||r|}
\hline\noalign{\smallskip}
Algorithm & \multicolumn{1}{|c|}{95\% CI} & \multicolumn{1}{|c||}{ACT} & \multicolumn{1}{|c|}{MSE Ratio} \\
\noalign{\smallskip}\hline\noalign{\smallskip}
CGS              & -0.0004 $\pm$ 0.0034 &  2.996 &   \multicolumn{1}{|c|}{1} \\
\noalign{\smallskip}\hline\noalign{\smallskip}
\multicolumn{1}{|r||}{RQGS $\;\; q=0.10$} & 0.0002   $\pm$ 0.0035 & 3.103 & 1.060 (0.068) \\
\multicolumn{1}{|r||}{$q=0.25$} & 0.0010  $\pm$ 0.0034  & 3.068 & 1.133 (0.070) \\
\multicolumn{1}{|r||}{$q=0.50$} & -0.0027  $\pm$ 0.0039 & 4.006 & 1.264 (0.078) \\
\multicolumn{1}{|r||}{$q=0.75$} & -0.0003 $\pm$ 0.0036 & 3.328 & 1.034 (0.063) \\
\multicolumn{1}{|r||}{$q=0.90$} & -0.0024 $\pm$ 0.0034  & 2.935 & 1.050 (0.063) \\
\noalign{\smallskip}\hline\noalign{\smallskip}
\multicolumn{1}{|r||}{RSGS $\;\;p= 0.10$} & 0.0029 $\pm$ 0.0063 & 20.274 & 3.831 (0.228) \\
\multicolumn{1}{|r||}{$p=0.25$} & 0.0024 $\pm$ 0.0047 & 11.668 & 2.094 (0.131) \\
\multicolumn{1}{|r||}{$p=0.50$} & -0.0005   $\pm$ 0.0045 & 10.350 & 2.063 (0.126) \\
\multicolumn{1}{|r||}{$p=0.75$} & -0.0030 $\pm$ 0.0059 & 17.983 & 3.077 (0.184) \\
\multicolumn{1}{|r||}{$p=0.90$} & 0.0004 $\pm$ 0.0087 & 39.203 & 7.022 (0.440) \\
\noalign{\smallskip}\hline
\end{tabular}
\end{table}



\subsection{A Bayesian General Linear Model}

\cite{john:jone:2010} establish geometric ergodicity for the CGS for a popular Bayesian general linear model.  Thus by Theorem \ref{thm:all},  the RQGS and RSGS are also geometrically ergodic.  An inspection of their proofs shows that the authors establish these results using the same techniques as those outlined by Theorem \ref{thm:recipes}.  For ease of exposition, we illustrate this approach for a (very) special case of this model, a Bayesian balanced random intercept model for $K$ subjects with $M$ observations on each.  
Specifically, let $Y$ denote an $N \times 1$ response
vector, $\beta$ a $p \times 1$ vector of regression
coefficients, and $u$ a $K \times 1$ vector.  Further, let $X$ be an $N \times
p$ design matrix of full column rank and $Z = I_K \otimes 1_M$ where $\otimes$ denotes the Kronecker product and $1_M$ is an $M \times 1$ vector of ones. Then the model is 
\begin{equation}\label{eq:hier}
\begin{split}
Y|\beta,u,\lambda_R,\lambda_D  & \sim \text{N}_N\left(X\beta + Zu,
\lambda_R^{-1}I_{N}\right) \\ 
\beta| u,\lambda_R,\lambda_D   & \sim
\text{N}_p\left(0,I_p\right)\\   
u|\lambda_R,\lambda_D          & \sim
\text{N}_K\left(0,\lambda_D^{-1}I_{K}\right) \\ 
\lambda_R     & \sim \text{Gamma}\left(2,1\right) \\ 
\lambda_D     & \sim \text{Gamma}\left(2,1\right) \\ 
\end{split}
\end{equation}
where we say $W \sim \text{Gamma}(2,1)$ if it has density proportional
to $we^{-w}$ for $w > 0$.  We also assume that $\beta$ and $u$
are conditionally independent given $\lambda_R$,
$\lambda_D$, and $y$ (ie.~$X^TZ=0$).

We can explore the posterior distribution of $\beta$, $u$, $\lambda_R$, and $\lambda_D$ given data $y$  using a two-component GS with components $\xi=\left(u^T,\beta^T\right)^T$ and $\lambda = \left(\lambda_R,
  \lambda_D \right)^T$.  Constructing the corresponding Markov chain $\Phi = \lb\left(\lambda^{(i)},\xi^{(i)} \right) \rb_{i=0}^\infty$ requires draws from the following full conditional distributions.  Letting $v_1(\xi) = (y-X\beta-Zu)^T(y-X\beta-Zu)$ and $v_2(\xi) = u^Tu$,
\[
\lambda | \xi, y  \sim \text{Gamma}\left(2 + \frac{N}{2}, \; 1 +  \frac{1}{2}v_1(\xi)\right)  \cdot \text{Gamma}\left(2 + \frac{K}{2}, \; 1 +  \frac{1}{2}v_2(\xi)\right) \; .
\]
That is, the conditional distribution of $\lambda=(\lambda_R,\lambda_D)$ given $(\xi,y)$ is the product of two independent Gamma distributions.  Further,
\[
\xi | \lambda, y  \sim N_{K+p}\left(\mu, \Sigma^{-1} \right) 
\]
where
\begin{equation}\label{eq:norm1}
\begin{split}
\Sigma^{-1} & = \left( \begin{array}{cc}
\left(\lambda_RM+ \lambda_D\right)^{-1}I_K & 0 \\
0               & \left(\lambda_RX^TX+I_p\right)^{-1} \\
\end{array} \right) \;\; \text{ and } \\
\mu & = \lambda_R \Sigma^{-1} \left( \begin{array}{c}
Z^Ty \\
X^Ty \\ 
\end{array} \right) \; .\\
\end{split}
\end{equation}
Accordingly, the GS and posterior density satisfy {\it Assumption $\mathcal{A}$} and the following Lemma establishes the sufficient conditions required by Theorem \ref{thm:recipes}.  Geometric ergodicity of the CGS, RQGS, and RSGS follows.  Please see \cite{john:jone:2010} for a proof of the Lemma.
\begin{lemma}\label{lem:bayes}
Define
\[
\begin{split}
f(\lambda) & = K \left(\frac{1}{\lambda_R} + \frac{1}{\lambda_D} \right) + e^{(\lambda_R+\lambda_D)/2} + 1 \\
g(\xi) & = v_1(\xi) + v_2(\xi) + 1  \; .\\
\end{split}
\]
These functions satisfy \eqref{eq:fg}  with $j = K/(2+K)$, $m=1$, and 
\[
\begin{split}
k & = \frac{2K}{N+2} + \frac{2K}{K+2} + 2^{K/2+2} + 2^{N/2+2} \\
n & = \sum_{i=1}^N x_i x_i^T + y^T\left(I_N + \frac{1}{M^2}ZZ^T\right)y \\
\end{split}
\]
where $x_i$ denotes the $i$th row of $X$.  Further, $C_d := \{\xi: g(\xi) \le d\}$ is compact for all $d>0$.  

\end{lemma}

Under geometric ergodicity, inference for the posterior can be guided by the existence of a CLT and consistent estimates of Monte Carlo standard errors.  For details, examples, and further study of the convergence among the GS for this model, please see \cite{john:jone:2010} and \cite{john:jone:neat:2013}.

\section{Appendix}

\subsection{Preliminaries}

The following lemmas are applied extensively throughout the appendix.  The first provides the notation and structure required for constructing CGS, RQGS, and RSGS drift conditions.

\begin{lemma}\label{lem:operator}
Denote expectation with respect to the full conditional distributions as
\[
\begin{split}
E[f(x,y)|x] & : = \int f(x,y) \pi(y|x) \mu_y(dy) \\
E[f(x,y)|y] & : = \int f(x,y) \pi(x|y) \mu_x(dx) \\
\end{split} \;.
\]
Then for any function $f(x,y)$, 
\[
\begin{split}
P_{CGS}f(x,y) & = E\left[E\left[f(x',y') \; \vline \; x'\right] \;\vline \; y \right] \\
P_{RQGS,q}f(x,y) & = qE\left[E\left[f(x',y') \; \vline \; x'\right] \;\vline \; y \right] + (1-q) E\left[E\left[f(x',y') \; \vline \; y'\right] \;\vline \; x \right] \\
P_{RSGS,p}f(x,y) & = p E\left[f(x',y) \; \vline \; y \right] + (1-p)E\left[f(x,y') \; \vline \; x \right] \; .\\
\end{split}
\]
\end{lemma}
\begin{proof}
First, it follows from \eqref{eq:operator} that
\[
\begin{split}
P_{CGS}f(x,y) 
& = \int \int f(x',y') \pi(x'|y)\pi(y'|x') \mu_x(dx')\mu_y(dy') \\
& = \int \left[\int f(x',y') \pi(y'|x') \mu_y(dy') \right] \pi(x'|y) \mu_x(dx') \\
& = E\left[E\left[f(x',y') \; \vline \; x'\right] \;\vline \; y \right] \\
\end{split}
\]
and the RQGS proof is similar. 
Finally, 
\[
\begin{split}
 P_{RSGS,p}f(x,y) 
 & =  p \int \int f(x',y')\pi(x'|y)\delta(y'-y)\mu_x(dx')\mu_y(dy')  \\
 & \hspace{.25in} + (1-p)\int \int f(x',y')\pi(y'|x)\delta(x'-x)\mu_x(dx')\mu_y(dy') \\
 & = p \int f(x',y)\pi(x'|y)\mu_x(dx') + (1-p) \int f(x,y')\pi(y'|x)\mu_y(dy') \\
& = p E\left[f(x',y) \; \vline \; y \right] + (1-p)E\left[f(x,y') \; \vline \; x \right] \; .\\
\end{split}
\]

\end{proof}

We will use the following Lemma to establish that our drift functions are unbounded off compact sets.
\begin{lemma}\label{lem:compact}
Suppose {\it Assumption $\mathcal{A}$} holds and function $V: \mathsf{X} \times \mathsf{Y} \to [1,\infty)$ is unbounded off compact sets.  
Then $\tilde{V}_1: \mathsf{X}\times\mathsf{Y} \to [1,\infty)$ and  $\tilde{V}_2: \mathsf{X}\times\mathsf{Y} \to [1,\infty)$ are also unbounded off compact sets where
\[
\begin{split}
\tilde{V}_1(x,y) & = u E[V(x,y')| x] + v E[ V(x',y)|y] \\
\tilde{V}_2(x,y) & = u E[E[V(x',y')|y']|x] + v E[E[V(x',y')|x']|y] + wE[V(x',y)|y]\\
\end{split}
\]
for $u,v,w \ge 0$ such that $u+v+w > 0$.
\end{lemma}
\begin{proof}

Let $C_d := \{ (x,y):  V(x,y) \le d \}$ where, by assumption, $C_d$ is compact for all $d>0$.  To establish that $\tilde{V}_1$ is unbounded off compact sets, we will prove that the $D_d  :=  \{(x,y): \tilde{V}_1(x,y) \le d\} $
is also compact (ie.~closed and bounded) for all $d>0$.  

First, we show $D_{d}$ is closed.  Specifically, we show that if $\{(x_i,y_i)\}_{i=1}^\infty \subseteq D_d$ and $\lim_{n \to \infty} (x_n,y_n) = (x,y)$, then $(x,y)$ is also in $D_d$ (ie.~$\tilde{V}_1(x,y) \le d$).
To this end, notice that for all $(\tx,\ty) \in \mathsf{X}\times\mathsf{Y}$
\[
V(x,\ty) \le \liminf_{n\to\infty} V(x_n,\ty)  
\hspace{.2in} \text{ and } \hspace{.2in}
V(\tx,y) \le \liminf_{n\to\infty}V(\tx,y_n)  
\]
by the closedness of $C_d$ and, by {\it Assumption $\mathcal{A}$},
\[
\pi(\tilde{y}|x) \le \liminf_{n \to \infty} \pi(\ty|x_n)
\hspace{.2in} \text{ and } \hspace{.2in}
\pi(\tilde{x}|y) \le \liminf_{n \to \infty} \pi(\tx|y_n) \; .
\]
Thus
\[
\begin{split}
\tilde{V}_1(x,y) 
 & = u E[V(x,\ty)| x] + v E[ V(\tx,y)|y] \\
& =  u\int V(x, \ty) \pi(\ty|x) \mu_y(d\ty) + v \int V(\tx, y) \pi(\tx | y) \mu_x(d\tx) \\
& \le u\int \liminf_{n\to\infty}V(x_n, \ty) \liminf_{n\to\infty}\pi(\ty|x_n) \mu_y(d\ty) + v \int \liminf_{n\to\infty}V(\tx, y_n)\liminf_{n\to\infty} \pi(\tx | y_n) \mu_x(d\tx) \\
& \le  u\int \liminf_{n\to\infty}V(x_n, \ty) \pi(\ty|x_n) \mu_y(d\ty) + v \int \liminf_{n\to\infty}V(\tx, y_n) \pi(\tx | y_n) \mu_x(d\tx) \\
& \le \liminf_{n\to\infty}\left[ u\int V(x_n, \ty) \pi(\ty|x_n) \mu_y(d\ty) + v \int V(\tx, y_n) \pi(\tx | y_n) \mu_x(d\tx)\right] \\
& = \liminf_{n\to\infty} \tilde{V}_1(x_n,y_n) \\
& \le d \\
\end{split}
\]
where the penultimate inequality follows from Fatou's lemma and the final inequality is guaranteed by $\{(x_i,y_i)\}_{i=1}^\infty \subseteq D_d$.  

Next, we show $D_d$ is bounded.  To this end, consider two cases.  First, suppose $V(x,y)$ attains a finite maximum value $m = \max_{(x,y) \in \mathsf{X}\times\mathsf{Y}}\{V(x,y)\}<\infty$.  In this case, $\mathsf{X} \times \mathsf{Y} = \{(x,y): V(x,y) \le m\} = C_m$
which is bounded by assumption.  Since $D_d \subset \mathsf{X} \times \mathsf{Y}$ for all $d$, $D_d$ must also be bounded.
On the other hand, suppose $V(x,y)$ does {\it not} attain a finite maximum.  In this case, 
define
\[
A_{(x,y)} = \{\tilde{x}: V(\tilde{x},y) \ge V(x,y)\} 
\hspace{.2in}
\text{ and }
\hspace{.2in}
B_{(x,y)} = \{\tilde{y}: V(x, \tilde{y}) \ge V(x,y)\} 
\]
and notice that for any $(x,y) \in \mathsf{X}\times\mathsf{Y}$, either $A_{(x,y)}$ has positive measure with respect to $\mu_x(d\tx)$ or $B_{(x,y)}$ has positive measure with respect to $\mu_y(d\ty)$.  Thus, for all $(x,y)$
\[
\begin{split}
\tilde{V}_1(x,y) 
& = u\int V(x, \ty) \pi(\ty|x)  \mu_y(d\ty)+ v \int V(\tx, y) \pi(\tx | y)  \mu_x(d\tx) \\
& \ge u\int_{B_{(x,y)}} V(x, \ty) \pi(\ty|x) \mu_y(d\ty)  + v \int_{A_{(x,y)}}V(\tx, y) \pi(\tx | y)  \mu_x(d\tx) \\
& \ge V(x,y) \left[ u\int_{B_{(x,y)}} \pi(\ty|x) \mu_y(d\ty)  + v \int_{A_{(x,y)}}\pi(\tx | y)  \mu_x(d\tx) \right] \\
& \ge  c V(x,y)   \\
\end{split}
\]
where $c :=  \min_{(x,y)}\left\lbrace u \int_{B_{(x,y)}} \pi(\ty|x) \mu_y(d\ty)  + v \int_{A_{(x,y)}}\pi(\tx | y) \mu_x(d\tx) \right\rbrace > 0$.
It follows that for all $(x,y) \in D_d$, $(x,y)$ is also in $C_{d/c}$ since
\[
V(x,y) \le \frac{\tilde{V}_1(x,y)}{c} \le \frac{d}{c} \; .
\]
Thus, $D_d \subseteq C_{d/c}$ and the boundedness of $D_d$ follows from the boundedness of $C_{d/c}$.  The proof that $\tilde{V}_2$ is unbounded off compact sets is similar, thus eliminated here.

\end{proof}

\subsection{Proof of Lemma \ref{lem:feller}}

We prove here that CGS is Feller.  The proofs for RQGS and RSGS are similar, thus eliminated.  First, 
\eqref{eq:c} guarantees that for any $(x,y), (x_n,y_n), (x_n',y_n') \in \mathsf{X}\times\mathsf{Y}$
\[
\begin{split}
 \liminf_{(x_n,y_n) \to (x,y)}  k_{CGS}((x_n,y_n),(x_n',y_n')) & =  \liminf_{(x_n,y_n) \to (x,y)}   \pi(x_n'|y_n) \pi(y_n'|x_n') \\
 &  \ge  \pi(x_n'|y) \pi(y_n'|x_n') \\
 & = k_{CGS}((x,y),(x_n',y_n')) \; .\\
 \end{split}
\]
Thus the CGS is Feller since for any open set $O \in \mathcal{B}$, an application of Fatou's Lemma shows that
\[
\begin{split}
\liminf_{(x_n,y_n) \to (x,y)} P_{CGS}((x_n,y_n), O) 
& = \liminf_{(x_n,y_n) \to (x,y)}  \iint\limits_{O} k_{CGS}((x_n,y_n),(x_n',y_n'))\mu_x(dx_n')\mu_y(dy_n') \\
& \ge  \iint\limits_{O} \liminf_{(x_n,y_n) \to (x,y)}  k_{CGS}((x_n,y_n),(x_n',y_n'))\mu_x(dx_n')\mu_y(dy_n') \\
& \ge  \iint\limits_{O} k_{CGS}((x,y),(x_n',y_n'))\mu_x(dx_n')\mu_y(dy_n') \\
& = P_{CGS}((x,y), O)  \; . \\
\end{split}
\]

\subsection{Proof of Theorem \ref{thm:Q}}

Geometric ergodicity of RQGS with sequence selection probability $q$ guarantees the existence of drift function $V: \mathsf{X} \times \mathsf{Y} \to [1, \infty)$, $\lambda \in (0,1)$, and finite constant $b > 0$ such that $V$ is unbounded off compact sets and 
\begin{equation}\label{proof:Q1}
\begin{split}
 P_{RQGS, q}V(x,y)
& = qE[E[V(x',y')|x']| y ] +  (1-q)E[E[V(x',y')|y']| x ]  \le \lambda V(x,y) + b 
\end{split}
\end{equation}
where the equality follows from Lemma \ref{lem:operator}.
Without loss of generality, assume $\lambda > \max\{q,1-q\}$ since if \eqref{proof:Q1} holds for $\lambda \le \max\{q,1-q\}$, it must also hold for $\lambda > \max\{q,1-q\}$.

To establish geometric ergodicity for CGS, define
\[
\begin{split}
g(x) & = E[E[V(x',y')|y']| x ] \\
h(y) & = E[E[V(x',y')|x']| y ] \\
z(y) & = E[V(x', y)|y] \\
\end{split}
\]
and constants $v$ and $w$ such that 
\[
\frac{\lambda(\lambda - q)}{q^2} < w < \frac{\lambda(1-q)}{q^2}
\;\;\;\; \text{ and } \;\;\;\;
\frac{\lambda(1-q)}{q^2} - w < v < \frac{\lambda(1-q) - (\lambda  -(1-q))qw}{\lambda q} \; .
\]
Also, define $\tilde{V}(x,y) =  v g(x) + h(y) + wz(y) \ge 1$.  By Lemma \ref{lem:compact}, $\tilde{V}$ is unbounded off compact sets.  Thus, geometric ergodicity of the CGS will follow from establishing the CGS drift condition
\[
P_{CGS}\tilde{V}(x,y) \le \tilde{\lambda} \tilde{V}(x,y) + \tilde{b} 
\]
for $\tilde{b} = b(v+w)/(1-q)$ and 
\[
 \max \lb \frac{q+\lambda}{q} - (v+w) \frac{q}{1-q}, \; \; \frac{\lambda}{w}\left( \frac{v+w}{1-q} -\frac{1}{q} \right)  \rb  < \tilde{\lambda} < 1 \;.
\]
To this end,  first notice that the RQGS drift condition guarantees 
\[
qh(y) + (1-q)g(x) = P_{RQGS, q}V(x,y) \le \lambda V(x,y) + b \; .
\]  
In conjunction with Lemma \ref{lem:operator}, it follows that 
\[
\begin{split}
 P_{CGS}g(x)= E[E[g(x')|x'] | y] 
 & = E[g(x') | y] \\
 & = \frac{1}{1-q}E \left[(1-q)g(x') +  qh(y)  \; \vline \; y  \right] - \frac{q}{1-q}h(y) \\
 & \le \frac{1}{1-q}E [\lambda V(x',y) + b |y] - \frac{q}{1-q}h(y)\\
 & =  \frac{\lambda}{1-q} z(y) -   \frac{q}{1-q}h(y) + \frac{1}{1-q}b \; ,\\
 \end{split}
\]
\[
\begin{split}
 P_{CGS}h(y)
 & = E[E[h(y')|x'] | y] \\
 & = \frac{1}{q}E\left[  E \left[qh(y') +  (1-q)g(x')  |x'\right] \; \vline \; y  \right] -  \frac{1-q}{q}E\left[E \left[ g(x') | x'\right] \; \vline \; y  \right] \\
 & \le \frac{1}{q}E\left[  E [\lambda V(x',y') + b|x']  \; \vline \; y  \right]  -  \frac{1-q}{q}E\left[E \left[ g(x') | x'\right] \; \vline \; y  \right]  \\
 & =  \frac{\lambda}{q} h(y) -  \frac{1-q}{q}P_{CGS}g(x)  +  \frac{1}{q}b\\
 \end{split}
\]
and
\[
\begin{split}
 P_{CGS}z(y) 
 & = E[E[z(y')|x'] | y] \\
 & = E[E[ E[V(\tx, y')|y']|x'] | y] \\ 
 & = E[g(x') | y] \\ 
 & = E[E[g(x') |x']| y] \\ 
 & = P_{CGS}g(x) \; . \\
 \end{split}
\]
Thus
\[
\begin{split}
 P_{CGS}\tilde{V}(x,y)
 & = v P_{CGS}g(x) + P_{CGS}h(y) + w P_{CGS}z(y) \\
 & \le \left(v + w - \frac{1-q}{q}\right) P_{CGS}g(x) + \frac{\lambda}{q} h(y) + \frac{1}{q}b \\
 & \le  \left( \frac{q+\lambda}{q} - (v+w) \frac{q}{1-q}\right) h(y)  +  \frac{\lambda}{w}\left( \frac{v+w}{1-q} -\frac{1}{q} \right) wz(y) + \frac{v+w}{1-q}b \\
 & \le \tilde{\lambda}(h(y) + wz(y)) + \tilde{b} \\
 & \le \tilde{\lambda} \tilde{V}(x,y)  + \tilde{b} \\
 \end{split} 
\]
and the result holds.

\subsection{Proof of Theorem \ref{thm:S}}

Geometric ergodicity of RSGS with component selection probability $p$ guarantees the existence of drift function $V: \mathsf{X} \times \mathsf{Y} \to [1, \infty)$, $\lambda \in (0,1)$, and finite constant $b > 0$ such that $V$ is unbounded off compact sets and 
\[
\begin{split}
 P_{RSGS, p}V(x,y)
& = pE[V(x',y)|y] + (1-p)E[V(x,y')|y] \le \lambda V(x,y) + b 
\end{split}
\]
where the equality holds from Lemma \ref{lem:operator}.
Without loss of generality, we assume $\lambda > \max\{p,1-p\}$ (see the proof of Theorem \ref{thm:Q}).

To establish geometric ergodicity for CGS, define
\[
g(x)  = E[ V(x,y') |x] 
\hspace{.25in} \text{ and } \hspace{.25in} 
h(y)  = E[V(x',y)|y]
\]
and constant 
\[
v > \frac{p(\lambda-p)}{\lambda(1-\lambda)} \; .
\]
Also, define $\tilde{V}(x,y) = g(x) + vh(y) \ge 1$.  It follows from Lemma \ref{lem:compact} that $\tilde{V}$ is unbounded off compact sets.   We will also show that $\tilde{V}$ satisfies the CGS drift condition 
\[
P_{CGS}\tilde{V}(x,y) \le \tilde{\lambda} \tilde{V}(x,y) + \tilde{b} 
\]
for
 \[
\frac{\lambda -p}{v(1-p)} + \frac{(\lambda -p)(\lambda+p-1)}{p(1-p)} \le \tilde{\lambda} < 1
 \]
and $\tilde{b} = (\lambda v + p)/(p(1-p))$.
Geometric ergodicity of the CGS will follow.

First, notice that the RSGS drift condition guarantees
\[
ph(y) + (1-p)g(x) = P_{RSGS, p}V(x,y) \le \lambda V(x,y) + b \; .
\]  
Thus 
\[
\begin{split}
 P_{CGS}g(x)  = E[E[g(x')|x'] | y] 
 & = E[ g(x') | y] \\
 & =  \frac{1}{1-p}E\left[ (1-p)g(x') + ph(y) \; \vline \; y  \right]  - \frac{p}{1-p}h(y) \\ 
 & \le  \frac{1}{1-p}E\left[\lambda  V(x',y) +b \; \vline \; y  \right] - \frac{p}{1-p}h(y) \\  
 & =  \frac{\lambda -p}{1-p}h(y)    + \frac{b}{1-p} \\  
 \end{split}
\]
and
\[
\begin{split}
 P_{CGS}h(y) = E[E[h(y')|x'] | y] 
 & = E\left[\frac{1}{p}E\left[ph(y') + (1-p)g(x')\;\vline\;x'\right]  - \frac{1-p}{p}g(x') \; \vline \; y\right] \\
 & \le E\left[  \frac{1}{p}E \left[\lambda V(x',y') + b \; \vline \; x'\right] - \frac{1-p}{p}g(x') \; \vline \; y  \right] \\
 & =   \frac{\lambda + p-1}{p} E\left[ g(x') \; \vline \; y  \right] + \frac{b}{p}\\ 
 & =   \frac{\lambda +p-1}{p(1-p)} E\left[ (1-p)g(x') + ph(y) \; \vline \; y  \right]  - \frac{\lambda +p-1}{1-p}h(y)+ \frac{b}{p}\\ 
 & \le \frac{\lambda +p-1}{p(1-p)}E\left[ \lambda V(x',y) + b \; \vline \; y  \right] -  \frac{\lambda +p-1}{1-p}h(y)  + \frac{b}{p}\\  
 & =  \frac{(\lambda -p)(\lambda+p-1)}{p(1-p)}h(y)    + \frac{\lambda}{p(1-p)}b \; .\\  
 \end{split}
\]
Combining these results establishes the CGS drift condition:
\[
\begin{split}
 P_{CGS}\tilde{V}(x,y) & = P_{CGS}g(x) + vP_{CGS}h(y)\\
 & \le \left[\frac{\lambda -p}{v(1-p)} + \frac{(\lambda -p)(\lambda+p-1)}{p(1-p)}\right] vh(y) + \left(\frac{\lambda v+p}{p(1-p)}\right)  b \\
 & \le \left[\frac{\lambda -p}{v(1-p)} + \frac{(\lambda -p)(\lambda+p-1)}{p(1-p)}\right] [g(x) + vh(y)] + \left(\frac{\lambda v+p}{p(1-p)}\right)  b \\
  & \le \tilde{\lambda} \tilde{V}(x,y) + \tilde{b} \; . \\
 \end{split}
\]

\subsection{Proof of Theorem \ref{thm:recipes}}

First, consider the CGS Markov chain $\Phi := \lb \left(X^{(0},Y^{(0)} \right), \left(X^{(1},Y^{(1)} \right), \left(X^{(2},Y^{(2)} \right), \ldots \rb$ and its $y$ {\it sub-chain} $\Phi_y:= \lb Y^{(0)},Y^{(1)},Y^{(2)}, \ldots \rb$ with Mtd and Markov kernel
\[
\begin{split}
k_{CGS,y}(y,y') & = \int \pi(x'|y)\pi(y'|x') \mu_x(dx') =  \int k_{CGS}((x,y),(x',y'))\mu_x(dx') \\
P_{CGS,y}(y, A) & = \int_A k_{CGS,y}(y,y')\mu_y(dy') = P_{CGS}((x,y), \mathsf{X}\times A) \; . \\
\end{split}
\] 
Notice that for any $g:\mathsf{Y} \to \mathbb{R}$,
\[
\begin{split}
P_{CGS,y}g(y) & = \int g(y') k_{CGS,y}(y,y') \mu_y(dy') \\
& = \int \int g(y') k_{CGS}((x,y),(x',y')) \mu_x(dx')\mu_y(dy') \\
& = P_{CGS}g(y) \; .\\
\end{split}
\]
It is also well known that, in this two-component setting, if $\Phi_y$ is geometrically ergodic, so is $\Phi$ (Roberts and Rosenthal (2001)). 
Thus, it suffices to establish geometric ergodicity for $\Phi_y$.  To this end, let $V_{CGS}$, $\lambda_{CGS}$, and $b_{CGS}$ be as defined in Lemma \ref{lem:recipes}.  Then the following drift condition holds for both the CGS and its $y$ sub-chain:
\[
\begin{split}
P_{CGS,y}g(y) = P_{CGS}g(y) 
& = E[E[g(y')|x']|y] \\
& \le E[mf(x') + n|y] \\
& \le jmg(y) + (mk + n) \\
& = \lambda_{CGS} g(y) + b_{CGS} 
\end{split}
\]
where $g(y) = V_{CGS}(x,y)$.
Further, by the assumption that $C_d := \{y: g(y) \le d\}$ is compact for all $d>0$, $g(y)$ is unbounded off compact sets for $\Phi_y$.  Thus $\Phi_y$ (and $\Phi$) is geometrically ergodic.

In conjunction with Theorem \ref{thm:all}, geometric ergodicity of the CGS guarantees the same for RQGS and RSGS.  Though unnecessary, it is still interesting to note that the conditions of Theorem \ref{thm:recipes} can be used to construct drift conditions for RQGS and RSGS.  To this end,  let $V_{RQGS}$, $\lambda_{RQGS}$, $b_{RQGS}$, and $v = v_{RQGS,q}$ be as defined in Lemma \ref{lem:recipes}. Then the following RQGS drift condition holds:
\[
\begin{split}
P_{RQGS,q}V_{RQGS}(x,y) & = qE[E[V_{RQGS}(x',y')|x']|y] + (1-q)E[E[V_{RQGS}(x',y')|y']|x] \\
& = qE[E[f(x') + vg(y')|x']|y] + (1-q)E[E[f(x') + vg(y')|y']|x] \\
& \le qE[(1+vm)f(x') + vn|y] + (1-q)E[(j+v)g(y')+k|x] \\
& \le (1-q)(j+v)mf(x) + \frac{qj(1+vm)}{v}vg(y)  + b_{RQGS} \\
& = \lambda_{RQGS} V_{RQGS}(x,y) + b_{RQGS} \\
\end{split}
\]
where the final equality holds since $v$ is a solution to $(1-q)(j+v)m = qj(1+vm)/v$.
Further, defining $V_{RSGS}$, $\lambda_{RSGS}$, $b_{RSGS}$, and $v = v_{RSGS,p}$ by Lemma \ref{lem:recipes} produces the following RSGS drift condition:
\[
\begin{split}
P_{RSGS,p}V_{RSGS}(x,y) & = pE[V_{RSGS}(x',y)|y] + (1-p)E[V_{RSGS}(x,y')|x] \\
& = pE[f(x') + vg(y)|y] + (1-p)E[f(x) + vg(y')|x] \\
& \le (1-p)(1+vm)f(x) + \frac{p(j+v)}{v}vg(y) + pk + (1-p)vn\\
& = \lambda_{RSGS} V_{RSGS}(x,y) + b_{RSGS} \\
\end{split}
\]
where the final equality holds since $v$ is a solution to $(1-p)(1+vm) = p(j+v)/v$.

%



\bibliographystyle{spbasic}      

\bibliography{mcref}

\end{document}